\documentclass[lettersize,onecolumn]{IEEEtran}
\usepackage{amsmath,amsfonts}
\usepackage{algorithmic}
\usepackage{algorithm}
\usepackage{array}
\usepackage[caption=false,font=normalsize,labelfont=sf,textfont=sf]{subfig}
\usepackage{textcomp}
\usepackage{stfloats}
\usepackage{url}
\usepackage{verbatim}
\usepackage{graphicx}
\usepackage{cite}
\usepackage[table]{xcolor}
\usepackage{threeparttable}
\usepackage{diagbox}
\usepackage{multirow}
\usepackage{makecell}
\usepackage{amsthm}
\usepackage{bm}
\theoremstyle{plain}

\newtheorem{theorem}{Theorem}

\theoremstyle{definition}
\newtheorem{example}{Example}
\newtheorem{construction}{Construction}
\newtheorem{remark}{Remark}

\begin{document}
\linespread{1.2}

\title{Reducing The Sub-packetization Level of Optimal-Access Cooperative MSR Codes}

\author{Yaqian Zhang\IEEEauthorrefmark{1},~Jingke Xu\IEEEauthorrefmark{2}

\thanks{\IEEEauthorrefmark{1} Yaqian Zhang is with Ministry of Education Key Laboratory of Mathematics and Information Networks, School of Mathematical Sciences,
    Beijing University of Posts and Telecommunications, Beijing 100876, China, (e-mail: zhangyq@bupt.edu.cn).}
\thanks{\IEEEauthorrefmark{2} Jingke Xu is with School of Information Science and Engineering, Shandong Agricultural University, Tai'an 271018, China (e-mail: xujingke@sdau.edu.cn).}

}


\IEEEpubid{0000--0000~\copyright~2023 IEEE}

\maketitle

\begin{abstract}
 Cooperative MSR codes are a kind of storage codes which enable optimal-bandwidth repair of any $h\geq2$ node erasures in a cooperative way, while retaining the minimum storage as an $[n,k]$ MDS code. Each code coordinate (node) is assumed to store an array of $\ell$ symbols, where $\ell$ is termed as sub-packetization.
Large sub-packetization tends to induce high complexity, large input/output in practice.
  To address the disk IO capability, a cooperative MSR code is said to have optimal-access property, if during node repair, the amount of data accessed at each helper node meets a theoretical  lower bound. 
  
  In this paper, we focus on reducing the sub-packetization of optimal-access cooperative MSR codes with two erasures.
  At first, we design two crucial MDS array codes for repairing a specific repair pattern of two erasures with optimal access.
  Then, using the two codes as building blocks and by stacking up of the two codes for several times, we obtain an optimal-access cooperative MSR code with two erasures.
  The derived code has sub-packetization $\ell=r^{\binom{n}{2}-\lfloor\frac{n}{r}\rfloor(\binom{r}{2}-1)}$ where $r=n-k$, and it reduces $\ell$ by a fraction of $1/r^{\lfloor\frac{n}{r}\rfloor(\binom{r}{2}-1)}$ compared with  the state of the art ($\ell=r^{\binom{n}{2}}$). 

\end{abstract}


\section{Introduction}
In distributed storage systems (DSS), data is stored across many storage nodes where node failures may occur frequently.
To protect data from node failures, erasure codes are used in DSS.
Typically, $k$ data blocks are encoded to $n$ blocks using an erasure code with each encoded block stored in one storage node.
If a node fails, it is repaired by downloading data from $d(\leq n-1)$ surviving nodes ({\it helper nodes}).
Two important metrics for the node repair efficiency are the total amount of data downloaded ({\it repair bandwidth}) and the volume of data accessed at the helper nodes, where the former indicates the network usage and the latter characterizes the disk I/O cost.
In \cite{Dimakis2011}, Dimakis et al. gave a tradeoff between the storage overhead and repair bandwidth, where codes with parameters lying on this tradeoff curve are called {\it regenerating codes}.
One extreme point of regenerating codes are minimum storage regenerating (MSR) codes, which achieve minimum storage overhead and have been extensively studied in the literature \cite{Ramchandran2011,Kumar2011,Vardy2016,Sasidharan2015,Ye2016,Ye2016sub-,Li-2018,Liu2022}.

MSR codes can only deal with single erasures, Lately cooperative MSR codes are defined in \cite{Hu2010} to repair $h\geq2$ erasures simultaneously through a cooperative repair model.
In the cooperative repair model, $h$ newcomers for repairing $h$ erasures download data from $d$ helper nodes and exchange data among themselves.
Specifically, suppose node $i$ stores a vector $\bm{c}_i\in F^{\ell}$ for $i\in[n]$, where $F$ is a finite field. Let $\mathcal{H}\subseteq[n]$ with $|\mathcal{H}|=h$ be the set of failed nodes.
For each $i\in\mathcal{H}$, let $\mathcal{R}_i\subseteq[n]\setminus\mathcal{H}$ with $|\mathcal{R}_i|=d$ denote the set of helper nodes connected by node $i$. Then the cooperative repair process includes the following two phases:
\begin{itemize}
\item \textbf{Download phase.} For each $i\in\mathcal{H}$ and $j\in\mathcal{R}_i$, node $i$ downloads $\beta_1$ symbols from helper node $j$ by accessing $\alpha_{i,j}$ coordinates of node $j$'s storage data $\bm{c}_j=(c_{j,0},\ldots,c_{j,\ell-1})$.
\item \textbf{Collaboration phase.} For each $i\in\mathcal{H}$ and $i'\in\mathcal{H}\setminus\{i\}$, node $i$ downloads $\beta_2$ symbols from node $i'$.
\end{itemize}
Then each node $i\in\mathcal{H}$ should be able to recover its erased data $\bm{c}_i$ using the data downloaded and exchanged in both repair phases.
The total repair bandwidth is defined as $\gamma=h(d\beta_1+(h-1)\beta_2)$ symbols, and the total amount of data accessed is $\gamma_a=\sum_{i\in\mathcal{H}}\sum_{j\in\mathcal{R}_i}\alpha_{i,j}$ symbols.

Indeed, cooperative MSR (also MSR) codes belong to a subclass of MDS codes, known as MDS array codes \cite{Blaum1998}.
An $(n,k,\ell)$ MDS array code over a finite field $F$ is formed by a set of vectors $(\bm{c}_1,\ldots,\bm{c}_n)$, where each $\bm{c}_i\in F^{\ell}$ $i\in[n]$ is a row vector of length $\ell$.
$\ell$ is called the sub-packetization level.
It satisfies that any $k$ coordinates $\bm{c}_i$ can be seen as information coordinates and can reconstruct the whole codeword (termed as MDS property).
Each coordinate $\bm{c}_i\in F^{\ell}$ is stored in one storage node $i$ for $i\in[n]$.
For MDS array codes, it is shown in \cite{Hu2010,Ye2018} that the repair bandwidth and the amount of accessed data for cooperatively repairing $h$ nodes using $d$ helper nodes are lower bounded by
\begin{equation}\label{bw-bound}
\gamma\geq\frac{(d+h-1)hl}{d-k+h}, \ \ \ \gamma_a\geq\frac{dhl}{d-k+h}.
\end{equation}
If for any $h$-subset $\mathcal{H}\subseteq[n]$ and any $d$-subset $\mathcal{R}_i\subseteq[n]\setminus\mathcal{H}$, $\gamma$ meets (\ref{bw-bound}) with equality, then the MDS array code is exactly a cooperative MSR code. Moreover, when both $\gamma$ and $\gamma_a$ achieve (\ref{bw-bound}) with equality, the code is said to satisfy the optimal-access property and called an optimal-access cooperative MSR code.

In previous works, cooperative MSR codes are presented in \cite{Ye2018,Zhang2020-scalar,Zhang2020,Ye2020,Liu2023,Zhang2025,Li2025}, wherein  \cite{Zhang2020,Li2025} satisfy the optimal-access property and \cite{Li2025} is a binary code constructed based on the code structure in \cite{Zhang2020}.
The codes in \cite{Zhang2020,Li2025} both have sub-packetization $\ell=(d-k+h)^{\binom{n}{h}}$ which is extraordinarily large due to practical consideration.
In this paper, we focus on reducing the sub-packetization for  optimal-access cooperative MSR codes with two erasure. The main contribution is that we propose a novel method to build such code  with sub-packetization $\ell=r^{\binom{n}{2}-\lfloor\frac{n}{r}\rfloor(\binom{r}{2}-1)}$, which reduces $\ell$ by a fraction of $1/r^{\lfloor\frac{n}{r}\rfloor(\binom{r}{2}-1)}$ compared with that in \cite{Zhang2020} for $h=2,d=n-2$.
Our main technique includes the following two folds:
(1) Design two crucial constructions of MDS array codes $\mathcal{C}_{\mathrm{I}}$ and $\mathcal{C}_{\mathrm{II}}$, which can cooperatively repair one specific repair pattern of two erasures with optimal access.
(2) Adopt the two codes $\mathcal{C}_{\mathrm{I}}$ and $\mathcal{C}_{\mathrm{II}}$ as building blocks, and construct the optimal-access cooperative MSR code by stacking up of  the two codes $\mathcal{C}_{\mathrm{I}}$ and $\mathcal{C}_{\mathrm{II}}$ for several times.

The remaining of the paper is organized as follows.
Section \ref{sec0-building-block} designs the two types of MDS array code building blocks $\mathcal{C}_{\mathrm{I}}$ and $\mathcal{C}_{\mathrm{II}}$.
Section \ref{sec1-two-nodes} presents the code construction of optimal-access cooperative MSR codes with smaller sub-packetization.
Section \ref{conclusion} concludes the paper.

\section{Two types of MDS array code building blocks}\label{sec0-building-block}
\subsection{Notations}
Throughout this paper, we use $[n]$ to denote the set of integers $\{1,2,...,n\}$ for a positive integer $n$, and denote $[i,j]=\{i,i+1,...,j\}$ for two integers $i<j$.
Let $F$ denote a finite field.
 Given an $(n,k,\ell)$ MDS array code $\mathcal{C}$ over $F$, for each codeword $\bm{c}\in\mathcal{C}$, we write $\bm{c}=(\bm c_1,...,\bm c_n)$ where $\bm c_i=(c_{i,0},c_{i,1},\ldots,c_{i,\ell-1})\in F^{\ell}$ for $i\in[n]$. Each coordinate $\bm{c}_i$ is called a node.
Note that the bold letters, suc as $\bm{c}, \bm{c}_i$, etc. always denote row vectors. $\bm{c}^{\tau}$ denotes the transpose of $\bm{c}$.
Let $I_{\ell}$ represent the identity matrix with order $\ell$.

We will define an $(n,k,\ell)$ MDS array code $\mathcal{C}$ by giving its parity-check matrix $H$.
Specifically, write
\begin{equation}\label{def1}
H=\begin{pmatrix}
H_{1,1} &H_{1,2}& \cdots &H_{1,n}\\
H_{2,1} &H_{2,2}& \cdots &H_{2,n}\\
\vdots&\vdots&\ddots&\vdots \\
H_{r,1} &H_{r,2}& \cdots &H_{r,n}
\end{pmatrix}
\end{equation}
where $r=n-k$ and $H_{t,i}$ is an $\ell\times \ell$ matrix over $F$ for $t\in[r]$ and $i\in[n]$.
That is, $\mathcal{C}=\{(\bm c_1,...,\bm c_n)\in(F^\ell)^n:~H\cdot(\bm c_1,...,\bm c_n)^{\tau}=\bm{0}\}$.
The MDS property of $\mathcal{C}$ indicates that any $r$  of the $n$ column blocks of $H$ form an invertible $r\ell\times r\ell$ matrix, equivalently, any
$k$ nodes are able to recover the whole codeword.
In the following, we always denote $r=n-k$.

\subsection{Two types of MDS array code building blocks}

We construct two types of $(n,k,\ell=r)$  MDS array codes  which later serve as building blocks for constructing cooperative MSR codes with two erasures. 
The first code can cooperatively repair nodes $\{1,2\}$, while the second can repair any two failed nodes among the first $r$ nodes, both utilizing the remaining $n-2$ helper nodes.

\begin{construction}[Type-I MDS array code building block $\mathcal{C}_{\mathrm{I}}$]\label{construction-1}
{\it
Let $F$ be a finite field with $|F|>n+r-2$, and $\lambda_1,\lambda_2,\ldots,\lambda_n,\gamma_1,\ldots,\gamma_{r-2}$ be $n+r-2$ distinct elements in $F$.
The Type-I $(n,k,\ell=r)$ MDS array code $\mathcal{C}_{\mathrm{I}}$ is defined by the parity-check matrix $H$ with the form in (\ref{def1}), where for $t\in[r]$, $H_{t,j}=\lambda_{j}^{t-1}I_{\ell}$ for $j\in[3,n]$ and $(H_{t,1},H_{t,2})$ is defined as:
\begin{equation}\label{pc-block1}
\left(\begin{array}{ccccc|ccccc}
\lambda_1^{t-1} & 0 & \gamma_1^{t-1} &\ldots & \gamma_{r-2}^{t-1} & \lambda_2^{t-1} &  &  &&  \\
  & \lambda_1^{t-1} &  && & 0 & \lambda_2^{t-1} & \gamma_1^{t-1} &\ldots & \gamma_{r-2}^{t-1} \\
  & &  \lambda_1^{t-1} && &  &  & \lambda_2^{t-1} &  &  \\
    & &  & \ddots & &  &  &  & \ddots &  \\
    & &  &  & \lambda_1^{t-1} &  &  &  & & \lambda_2^{t-1} \\
\end{array}\right),
\end{equation}
where the empty positions represent zeros.
}
\end{construction}

We illustrate the MDS property and cooperative repair property of $\mathcal{C}_{\mathrm{I}}$ as follows.

\begin{itemize}
\item[1)] {\it MDS property.}

The MDS property of the code $\mathcal{C}_{\mathrm{I}}$ is straightforward. First, by the last $r-2$ rows of the matrix $(H_{t,1},\ldots,H_{t,n}), t\in[r]$, one can see that the punctured code of $\mathcal{C}_{\mathrm{I}}$ by deleting the first two symbols $c_{i,0}, c_{i,1}$ from each coordinate $\bm{c}_i$ for $i\in[n]$, i.e., $\{(c_{i,2},\ldots,c_{i,\ell-1})_{1\leq i\leq n}: (c_{i,0},\ldots,c_{i,\ell-1})_{1\leq i\leq n}\in \mathcal{C}_{\mathrm{I}}\}$ forms an $(n,k,r-2)$ MDS array code. Then, according to the first two rows of $(H_{t,1},\ldots,H_{t,n}), t\in[r]$ and substituting the punctured code into it, one can obtain the MDS property of $\mathcal{C}_{\mathrm{I}}$.

\item[2)]  {\it Cooperative repair of $\{1,2\}$. }

Suppose node $1$ and node $2$ are erased.
By the first row of $(H_{t,1},\ldots,H_{t,n}), t\in[r]$, one can obtain the following parity check equations
$$
\sum_{j\in[n]}\lambda_{j}^{t-1}c_{j,0}+\gamma_1^{t-1}c_{1,2}+\cdots+\gamma_{r-2}^{t-1}c_{1,r-1}=0, ~~~~t\in[r].
$$
This implies that $(c_{1,0}, c_{1,2},\ldots,c_{1,r-1},c_{2,0},\ldots,c_{n,0})$ constitutes an $[n+r-2, n-2]$ generalized Reed-Solomon code (GRS) codeword.
Then node $1$ by downloading the symbol $c_{j,0}$ from each helper node $j\in[3,n]$, can obtain the data $c_{1,0}, c_{1,2}, \dots, c_{1,r-1}$ and $c_{2,0}$.
Similarly, based on the second row of $(H_{t,1},\ldots,H_{t,n}), t\in[r]$, node $2$ can obtain $c_{2,1},c_{2,2},\ldots,c_{2,r-1}$ and $c_{1,1}$ by downloading the helper node data $\{c_{j,1}: j\in[3,n]\}$.
Obviously, the repair can be done by one more exchange of the symbols $c_{2,0}$ and $c_{1,1}$ between the two nodes.

\end{itemize}

\begin{construction}[Type-II MDS array code building block $\mathcal{C}_{\mathrm{II}}$]\label{construction-2}
{\it
Let $F$ be a finite field with $|F|>n$, and $\lambda_1,\lambda_2,\ldots,\lambda_n,\tau\in F$ such that $\lambda_i, i\in[n]$ are all distinct and $\tau\neq 0, 1$.
The Type-II $(n,k,\ell=r)$ MDS array code $\mathcal{C}_{\mathrm{II}}$ is defined by the parity-check matrix $H$ with the form in (\ref{def1}), where for $t\in[r]$, $H_{t,j}=\lambda_{j}^{t-1}I_{\ell}$ for $j\in[r+1,n]$ and $(H_{t,1},\dots,H_{t,r})=$ 
\begin{equation}\label{pc-block2}
\begin{aligned}
&\left(\begin{array}{ccccc|ccccc|c|ccccc}
\lambda_1^{t-1} & \lambda_2^{t-1} & \lambda_3^{t-1} &\ldots & \lambda_{r}^{t-1} & \lambda_2^{t-1} &  &  && & &
\lambda_r^{t-1} &  & & &  \\
 & \lambda_1^{t-1} &  && & \tau\lambda_1^{t-1} & \lambda_2^{t-1} & \lambda_3^{t-1} &\ldots & \lambda_{r}^{t-1} & &
 & \lambda_r^{t-1} & & &  \\
 & &  \lambda_1^{t-1} && &  &  & \lambda_2^{t-1} &  & & \cdots
 & & & \lambda_{r}^{t-1} & &  \\
 & &  & \ddots & &  &  &  & \ddots &  &
 & & & & \ddots & \\
 & &  &  & \lambda_1^{t-1} &  &  &  & & \lambda_2^{t-1} &
 & \tau\lambda_1^{t-1} & \tau\lambda_2^{t-1} & \tau\lambda_3^{t-1} & \ldots & \lambda_r^{t-1}\\
\end{array}\right)
\end{aligned}
\end{equation}
where the empty positions in (\ref{pc-block2}) represent zeros.
And note that in (\ref{pc-block2}) the element $\tau$ only appears in the $(i,1),\ldots,(i,i-1)$-th entries of $H_{t,i}$ for $i\in[r]$.
}
\end{construction}
Next we prove the MDS property and cooperative repair property of $\mathcal{C}_{\mathrm{II}}$.
\begin{itemize}

\item[1)] {\it MDS property.}

It suffices to show that any $r$ column blocks of $H$ form an invertible matrix.
Suppose that the $r$ column block indices are $i_1,i_2,\ldots,i_r\in[n]$ and $i_1<i_2<\ldots<i_r$.
Suppose $i_1,\ldots,i_g\in[r]$ and $i_{g+1},\ldots,i_r\in[r+1,n]$ for some $0\leq g\leq r$.
Denote by $H(i_1,i_2,\ldots,i_r)$ the submatrix of $H$ consisting of the $r$ column blocks. We prove that for any $\bm{x}\in(F^{r})^r$, $H(i_1,i_2,\ldots,i_r)\cdot \bm{x}^{\tau}=\bm{0}$ always implies $\bm{x}=\bm{0}$. Denote $\bm{x}=(\bm{x}_1,\ldots,\bm{x}_r)$ and $\bm{x}_i=(x_{i,1},\ldots,x_{i,r})$ for $i\in[r]$.

For each $a\in[r]\setminus\{i_1,\ldots,i_g\}$, according to the $a$-th row of $(H_{t,i_1},\ldots,H_{t,i_r})\bm{x}^{\tau}=\bm{0}$, $t\in[r]$, one can obtain that
$
\sum_{j=1}^{r}\lambda_{i_j}^{t-1}x_{j,a}=0, t\in[r].
$
This implies $x_{j,a}=0$ for all $j\in[r], a\in[r]\setminus\{i_1,\ldots,i_g\}$.

Now consider symbols $x_{j,i_s}$, $j\in[r],s\in[g]$. For each $s\in[g]$, by the $i_s$-th row of $(H_{t,i_1},\ldots,H_{t,i_r})\bm{x}^{\tau}=\bm{0}$, $t\in[r]$ and noticing that $x_{s,a}=0$ for all $ a\in[r]\setminus\{i_1,\ldots,i_g\}$, then one can obtain the parity check equations in (\ref{eq-1-0}) as follows: for $t\in[r]$,
\begin{equation}\label{eq-1-0}
\resizebox{0.8\width}{!}{$
\begin{aligned}
&\sum_{1\leq j<s}\lambda_{i_{j}}^{t-1}(x_{j,i_s}\!+\!\tau x_{s,i_{j}})\!+\!\sum_{s<j\leq g}\lambda_{i_{j}}^{t-1}(x_{j,i_s}\!+\!x_{s,i_{j}})  \\
+&\lambda_{i_{s}}^{t-1}x_{s,i_s}+\sum_{j=g+1}^{r}\lambda_{i_{j}}^{t-1}x_{j,i_s} =0.
\end{aligned} $}
\end{equation}
Thus, the following set of $r$ symbols:
$\{x_{s,i_s}\}\cup\{x_{j,i_s}+\tau x_{s,i_{j}}\}_{1\leq j<s}\cup\{x_{j,i_s}+x_{s,i_{j}}\}_{s<j\leq g}\cup\{x_{j,i_s}\}_{j\in[g+1,r]}$
are solved to be zeros by (\ref{eq-1-0}).
That is, for  $s\in[g]$, one can directly compute $x_{s,i_s}=0$ and $x_{j,i_s}=0,j\in[g+1,r]$, which implies that ${\bm x}_j={\bf 0}$ for all $j\in[g+1,r]$.
As for the symbol sums, by considering equations labeled by $i_{s}$ and $i_{j}$, one has
\begin{equation}\label{eq-1-1}
\begin{cases}
x_{j,i_s}+\tau x_{s,i_{j}}=0 \\
x_{s,i_{j}}+ x_{j,i_{s}}=0
\end{cases}
\mathrm{if}~j<s;~
\begin{cases}
x_{s,i_j}+ x_{j,i_{s}}=0 \\
x_{s,i_{j}}+\tau x_{j,i_s}=0
\end{cases}
\mathrm{if}~j>s.
\end{equation}
Since $\tau\neq0,1$ and $x_{s,i_s}=0$, it has that $x_{s,i_{j}}=0$ for all $s,j\in [g]$. Hence, ${\bm x}_s={\bf 0}$ for all $s\in[g]$. 
This completes the proof.

\item[2)] {\it Cooperative repair of any $\{j_1,j_2\}\subseteq[r]$.}

W.L.O.G., suppose the two nodes $\{j_1,j_2\}=\{1,2\}$ are erased.
By the first row of the matrix $(H_{t,1},\ldots,H_{t,n})$, $t\in[r]$, one obtains that
$(c_{1,0}, c_{1,1}+c_{2,0},\ldots,c_{1,r-1}+c_{r,0},c_{r+1,0},\ldots,c_{n,0})$ forms an $[n,k]$ GRS codeword.
Thus node 1 by downloading the symbols $c_{j,0}$, $j\in[3,n]$, can recover $c_{1,1}+c_{2,0}, c_{1,0}, c_{1,2},\ldots,c_{1,r-1}$.
Similarly, by the second row of $(H_{t,1},\ldots,H_{t,n}), t\in[r]$,
it has $(c_{1,1}+\tau c_{2,0}, c_{2,1}, c_{2,2}+c_{3,1}\ldots,c_{2,r-1}+c_{r,1},c_{r+1,1},\ldots,c_{n,1})$ forms an $[n,k]$ GRS codeword.
By downloading $c_{j,1}, j\!\in\![3,n]$, node 2 computes $c_{1,1}\!+\!\tau c_{2,0}, c_{2,1},\ldots,c_{2,r-1}$.
Then one more exchange between the two nodes can complete the repair since $\tau\neq0,1$.
Generally, any two nodes $\{j_1,j_2\}\!\subseteq\![r]$ can be cooperatively repaired by considering the $j_1$-th and $j_2$-th rows of the parity check equations.
\end{itemize}

Actually, we emphasize that $\mathcal{C}_{\mathrm{II}}$ can cooperatively repair in total $\binom{r}{2}$ erasure patterns. This good property of $\mathcal{C}_{\mathrm{II}}$ plays an important role on the reduction of sub-packetization in our cooperative MSR code construction. We will explain it later in Remark \ref{remark}.
Next, based on the two codes $\mathcal{C}_{\mathrm{I}}$ and $\mathcal{C}_{\mathrm{II}}$, we are able to construct an optimal-access cooperative MSR codes with  two erasures.

\section{optimal-access cooperative MSR codes with two erasures}\label{sec1-two-nodes}

In this section, we construct an optimal-access cooperative MSR code $\mathcal{C}$ for repairing two erasures.
Denote $m=\binom{n}{2}-\lfloor\frac{n}{r}\rfloor(\binom{r}{2}-1)$. The code $\mathcal{C}$ is an $(n,k,\ell=r^m)$ MDS array code over a finite field $F$, which is constructed by stacking up of the two MDS array code building blocks $\mathcal{C}_{\mathrm{I}}$ and $\mathcal{C}_{\mathrm{II}}$ for several times.
Before giving the construction, some notations and definitions are needed.

\subsection{Notations and definitions}\label{sec1-notation}

\begin{itemize}
\item 
Denote $m=\binom{n}{2}-\lfloor\frac{n}{r}\rfloor(\binom{r}{2}-1)$ and $\ell=r^m$. For an integer $a\in[0,\ell-1]$, $a$ has the unique $r$-ary expansion $a=\sum_{i=1}^{m}a_{i}r^{i-1}$ with $a_i\in[0,r-1]$ for $i\in[m]$. Then we write $a=(a_1,a_2,\ldots,a_m)$ for simplicity.
For some $i\in[m]$ and $v\in[0,r-1]$, denote $a(i,v)=(a_1,\ldots,a_{i-1},v,a_{i+1},\ldots,a_m)$.
\item
The $n$ storage nodes are indexed from $1$ to $n$.
Let $\mathcal{P}=\{(j,j'): 1\leq j<j'\leq n\}$ represent the set of all $\binom{n}{2}$ two-node pairs.
Define $\mathcal{P}_i=\{(j,j'): (i-1)r+1\leq j<j'\leq ir\}$ for $i=1,2,\ldots,\lfloor\frac{n}{r}\rfloor$, then $|\mathcal{P}_i|=\binom{r}{2}$.
Moreover, these $\mathcal{P}_i$, $i\in[\lfloor\frac{n}{r}\rfloor]$ are disjoint subsets of $\mathcal{P}$. Define $\mathcal{P}_0=\mathcal{P}\setminus(\cup_{i\in[\lfloor\frac{n}{r}\rfloor]}\mathcal{P}_i)$, then $|\mathcal{P}_0|=\binom{n}{2}-\lfloor\frac{n}{r}\rfloor\binom{r}{2}=m-\lfloor\frac{n}{r}\rfloor$.
\item
Define $\pi$ to be a surjective map from $\mathcal{P}$ to the set $[m]$ satisfying the following conditions.
\begin{itemize}
\item[(1)] For $i\in[\lfloor\frac{n}{r}\rfloor]$, $\pi$ maps pairs in $\mathcal{P}_i$ to integer $i$. That is, $\pi(j,j')=i$ if $(j,j')\in\mathcal{P}_i$ for $i=1,2,\ldots,\lfloor\frac{n}{r}\rfloor$.
\item[(2)] $\pi$ maps pairs in $\mathcal{P}_0$ to integers in $[\lfloor\frac{n}{r}\rfloor+1,m]$ and it is a one to one mapping. This can be done since $|\mathcal{P}_0|=m-\lfloor\frac{n}{r}\rfloor$.
\end{itemize}
\item
For each $j\in[n]$, define $\Omega_{j,0}$ and $\Omega_{j,1}$ to be the following two subsets of $[\lfloor\frac{n}{r}\rfloor+1,m]$:
$$\begin{aligned}
& \Omega_{j,0}=\{\pi(j,j_2): j<j_2\leq n~\mathrm{and}~(j,j_2)\in\mathcal{P}_0\}, \\
&\Omega_{j,1}=\{\pi(j_1,j): 1\leq j_1<j~\mathrm{and}~(j_1,j)\in\mathcal{P}_0\}.
\end{aligned}$$
\end{itemize}

\subsection{Code construction}

\begin{construction}\label{construction-3}
{\it
Let $F$ be a finite field with $|F|>n+r-2$. Let $\lambda_1,\lambda_2,\ldots,\lambda_n,\gamma_1,\gamma_2,\ldots,\gamma_{r-2}$ be distinct elements in $F$ and $\tau\in F\setminus\{0,1\}$.
Denote $m=\binom{n}{2}-\lfloor\frac{n}{r}\rfloor(\binom{r}{2}-1)$.
The $(n,k,\ell=r^m)$ cooperative MSR code $\mathcal{C}$ is defined by the parity-check matrix $H$ with the form in (\ref{def1}), where for $t\in[r]$ and $j\in[n]$, $H_{t,j}$'s are defined in (\ref{pc-two-nodes}) and (\ref{pc-two-nodes2}) as follows
(we index the rows and columns of $H_{t,j}$ by integers from $0$ to $\ell-1$).
\begin{itemize}
\item For \bm{$1\leq j\leq\lfloor\frac{n}{r}\rfloor r$}: Denote $j=(u-1)r+v+1$ for some $u\in[\lfloor\frac{n}{r}\rfloor]$ and $v\in[0,r-1]$.
For $0\leq a,b\leq\ell-1$, the $(a,b)$-th entry of $H_{t,j}$ is defined as $H_{t,j}(a,b)=$
\begin{equation}\label{pc-two-nodes}
\begin{cases}
\lambda_j^{t-1}~ ~~~~~~~~~~~~~~~~\mathrm{if}~a\!=\!b,\\
\tau\lambda_{(u-1)r+v'+1}^{t-1}~~~~~\mathrm{if}~a_{u}\!=\!v,b\!=a(u,v')~\mathrm{for}~v'\in[0,v\!-\!1], \\
\lambda_{(u-1)r+v'+1}^{t-1}~~~~~~~\mathrm{if}~a_{u}\!=\!v,b\!=\!a(u,v')~\mathrm{for}~v'\in[v\!+\!1,r\!-\!1], \\
\gamma_{w-1}^{t-1}~~~~~~~~~~~~~~~~\mathrm{if}~a_{i}\!=\!0,b\!=\!a(i,w)~\mathrm{for}~i\!\in\!\Omega_{j,0}, w\!\in\![2,r\!-\!1], \\
\gamma_{w-1}^{t-1}~~~~~~~~~~~~~~~~\mathrm{if}~a_{i}\!=\!1,b\!=\!a(i,w)~\mathrm{for}~i\!\in\!\Omega_{j,1}, w\!\in\![2,r\!-\!1], \\
0 ~~~~~~~~~~~~~~~~~~~~~\mathrm{otherwise}.
\end{cases}
\end{equation}
\item For \bm{$\lfloor\frac{n}{r}\rfloor r<j\leq n$}:
For $0\leq a,b\leq\ell-1$, the $(a,b)$-th entry of $H_{t,j}$ is defined as $H_{t,j}(a,b)=$
\begin{equation}\label{pc-two-nodes2}
\begin{cases}
\lambda_j^{t-1}~ ~~~~~~~~~~\mathrm{if}~~a\!=\!b,\\
\gamma_{w-1}^{t-1}~~~~~~~~~~\mathrm{if}~~a_{i}\!=\!0~\mathrm{and}~b\!=\!a(i,w)~\mathrm{for}~i\!\in\!\Omega_{j,0}, w\!\in\![2,r\!-\!1], \\
\gamma_{w-1}^{t-1}~~~~~~~~~~\mathrm{if}~~a_{i}\!=\!1~\mathrm{and}~b\!=\!a(i,w)~\mathrm{for}~i\!\in\!\Omega_{j,1}, w\!\in\![2,r\!-\!1], \\
0 ~~~~~~~~~~~~~~~\mathrm{otherwise}.
\end{cases}
\end{equation}
\end{itemize}
}
\end{construction}

\begin{remark}\label{remark}
Recall that the code $\mathcal{C}_{\mathrm{I}}$ and $\mathcal{C}_{\mathrm{II}}$ in Section \ref{sec0-building-block} can indeed repair one and $\binom{r}{2}$ repair patterns respectively.
In order to construct $\mathcal{C}$  with cooperative repair of any two erasures, i.e. $\binom{n}{2}$ repair patterns, we extend the code $\mathcal{C}_{\mathrm{I}}$ and $\mathcal{C}_{\mathrm{II}}$ to multiple dimensions.
That is, the sub-packetization of  $\mathcal{C}$  is extended as $\ell=r^m$.
In particular, we extend $\mathcal{C}_{\mathrm{II}}$ to $\lfloor\frac{n}{r}\rfloor$ dimensions for repairing $\lfloor\frac{n}{r}\rfloor\binom{r}{2}$ erasure patterns in $\cup_{u\in[\lfloor\frac{n}{r}\rfloor]}\mathcal{P}_u$,
and extend $\mathcal{C}_{\mathrm{I}}$ to $m-\lfloor\frac{n}{r}\rfloor$ dimensions for repairing the remaining $\binom{n}{2}-\lfloor\frac{n}{r}\rfloor\binom{r}{2}$ erasure patterns in $\mathcal{P}_0$, where $\mathcal{P}_u$'s, $\mathcal{P}_0$ are defined in Subsection \ref{sec1-notation}.
To this end, 
we construct each $H_{t,j}$ in (\ref{pc-two-nodes}), (\ref{pc-two-nodes2}) from the diagonal matrix $\lambda_j^{t-1}I_{\ell}$ by successively adding some non-diagonal entries, i.e., the symbols  $\tau\lambda_{(u-1)r+v'+1}^{t-1}$'s, $\lambda_{(u-1)r+v'+1}^{t-1}$'s similar to the matrix (\ref{pc-block2}) and $\gamma_{w-1}^{t-1}$'s similar to the matrix (\ref{pc-block1}).
These non-diagonal entries are used to execute node repair.
Indeed,  
for erasure patterns $(i_1,i_2)\in\cup_{u\in[\lfloor\frac{n}{r}\rfloor]}\mathcal{P}_u$, denote $\pi(i_1,i_2)=u$ and write $i_j=(u-1)r+v_j+1$, for $j\in[2]$. 
Then parity check rows labeled by $\{a: a_u=v_j\}$ are used by node $i_j$ for repairing, just as in code $\mathcal{C}_{\mathrm{II}}$.
For erasure patterns $(i_1,i_2)\in\mathcal{P}_0$, denote $\pi(i_1,i_2)=\rho$. 
Then parity check rows labeled by $\{a: a_\rho=0\}$ and $\{a: a_\rho=1\}$  are respectively used by node $i_1$ and node $i_2$ for repairing, just as in code $\mathcal{C}_{\mathrm{I}}$.
\end{remark}

For a better understanding of the construction, we give an illustrating example.

\begin{example}\label{ex}
Let $n=7, k=4, d=5, h=2$. Then $r=3$, $m=17$ and $\ell=3^{17}$.
The cooperative MSR code $\mathcal{C}$ is a $(7,4,3^{17})$ MDS array code.
Let $F$ be a finite field with $|F|\geq 9$. Choose $\lambda_1,\lambda_2,\ldots,\lambda_7,\gamma_1$ to be $8$ distinct elements in $F$ and choose $\tau\in F\setminus\{0,1\}$.
Every integer $a\in[0,\ell-1]$ is represented by a ternary vector $(a_1,a_2,\ldots,a_{17})$. The surjective map $\pi$ is displayed in Table \ref{tab1}.
 \begin{table}[ht]\label{tab1}
  \renewcommand{\arraystretch}{1.4}
	\centering
\caption{The surjective map $\pi$ on $\mathcal{P}=\mathcal{P}_0\cup\mathcal{P}_1\cup\mathcal{P}_2$}
\setlength{\tabcolsep}{1mm}{
\begin{tabular}{l|lllll}
\hline
$\pi$ on $\mathcal{P}_1$& $\pi(1,2)\!=\!1$& $\pi(1,3)\!=\!1$ & $\pi(2,3)\!=\!1$ & & \\ [0.8ex] \hline
$\pi$ on $\mathcal{P}_2$& $\pi(4,5)\!=\!2$ & $\pi(4,6)\!=\!2$ &$\pi(5,6)\!=\!2$ & & \\ [0.8ex] \hline
\multirow{3}{*}{$\pi$ on $\mathcal{P}_0$}&$\pi(1,4)\!=\!3$& $\pi(1,5)\!=\!6$  &$\pi(1,6)\!=\!9$ & $\pi(1,7)\!=\!12$ & $\pi(4,7)\!=\!15$ \\
 &$\pi(2,4)\!=\!4$ & $\pi(2,5)\!=\!7$ & $\pi(2,6)\!=\!10$ & $\pi(2,7)\!=\!13$ & $\pi(5,7)\!=\!16$ \\
 &$\pi(3,4)\!=\!5$  &$\pi(3,5)\!=\!8$ &$\pi(3,6)\!=\!11$ &$\pi(3,7)\!=\!14$ &$\pi(6,7)\!=\!17$
\\ \hline
\end{tabular}}
\end{table}

Take $j=2$ for example. Since $1\leq j=2\leq \lfloor\frac{n}{r}\rfloor r=6$, then according to (\ref{pc-two-nodes}), we construct $H_{t,2}$ by the following steps:
\begin{itemize}
\item Begin with the diagonal matrix $\lambda_2^{t-1}I_{3^{17}}$;
\item For $a\in[0,3^{17}-1]$ with $a_1=1$, set the $(a,a(1,0))$-th entry to be $\tau\lambda_{1}^{t-1}$, and set the $(a,a(1,2))$-th entry to be $\lambda_{3}^{t-1}$;
\item Since $\Omega_{2,0}=\{4,7,10,13\}$ and $\Omega_{2,1}=\emptyset$, then for $a\in[0,3^{17}-1]$ with $a_4=0$ (or $a_7=0$ or $a_{10}=0$ or $a_{13}=0$), set the $(a,a(4,2))$-th (or $(a,a(7,2))$-th or $(a,a(10,2))$-th or $(a,a(13,2))$-th) entry to be $\gamma_{1}^{t-1}$.
\end{itemize}
Note that for each $a\in[0,\ell-1]$, the $a$-th row of $H_{t,2}$ may have several non-diagonal entries being $\tau\lambda_{1}^{t-1}$, $\lambda_{3}^{t-1}$ or $\gamma_{1}^{t-1}$, and these non-diagonal entries added at each step all lie in different columns.

\end{example}

In the following, we show the MDS property and optimal-access property of the code $\mathcal{C}$ in Construction \ref{construction-3} in Subsection \ref{mds-h2} and Subsection \ref{repair-h2}, respectively.

\subsection{MDS property}\label{mds-h2}
We prove the MDS property of the code $\mathcal{C}$ in Construction \ref{construction-3}. 
For convenience, for every integer $a\in[0,\ell-1]$, define $w_{\rm{suf}}(a)$ to be the number of digits in $(a_{\lfloor\frac{n}{r}\rfloor+1}, \ldots, a_m)$ being $0$ or $1$.
That is, $w_{\rm{suf}}(a)=|\{u\in[\lfloor\frac{n}{r}\rfloor+1,m]: a_u=0~{\rm{or}}~1\}|$.
For $0\leq s\leq m-\lfloor\frac{n}{r}\rfloor$, define $\mathcal{L}_s=\{a\in[0,\ell-1]: w_{\rm{suf}}(a)=s\}$.
Then $\mathcal{L}_0, \mathcal{L}_1,\ldots,\mathcal{L}_{m-\lfloor\frac{n}{r}\rfloor}$ form a partition of $[0,\ell-1]$.
The definition of $\mathcal{L}_s$ will be used in the proof of both MDS property and optimal access property of $\mathcal{C}$. Next in Theorem \ref{thm-mds} we firstly prove the MDS property of $\mathcal{C}$.

\begin{theorem}\label{thm-mds}
The code $\mathcal{C}$ in Construction \ref{construction-3} has MDS property.
\end{theorem}

\begin{proof}
It suffices to prove that any $r$ column blocks $i_1,i_2,\ldots,i_r$ with $i_1<i_2<\cdots<i_r$ of $H$, denoted by $H(i_1,\ldots,i_r)$, forms an invertible matrix. Equivalently, we prove that for any $\bm{x}=(\bm{x}_1,\ldots,\bm{x}_r)$ with $\bm{x}_i=(x_{i,0},\ldots,x_{i,\ell-1})\in F^{\ell}$ for $i\in[r]$,  it has $H(i_1,\ldots,i_r)\cdot\bm{x}^{\tau}=\bm{0}$ always implies $\bm{x}=\bm{0}$.

Note that for each $1\leq i\leq r\lfloor\frac{n}{r}\rfloor$, we can write $i=(u-1)r+v+1$ for some $u\in[\lfloor\frac{n}{r}\rfloor]$ and $v\in[0,r-1]$. Denote $i=(u,v)$ for simplicity.
That is, the first $r\lfloor\frac{n}{r}\rfloor$ nodes are partitioned into $\lfloor\frac{n}{r}\rfloor$ groups, and node $i$ locates in the $(v+1)$-th node of the $u$-th group.
Then, for each $i\in\{i_1,i_2,\ldots,i_r\}$ with $i\leq r\lfloor\frac{n}{r}\rfloor$, $i$ can be written as $i=(u,v)$. Suppose these nodes locate in the groups $U=\{u_1,\ldots, u_e\}$ for some $0\leq e\leq r$. 
If $U\neq \emptyset$,
for each $j\in[e]$, denote $V_j=\{v\in[0,r-1]: (u_j,v)\in\{i_1,i_2,\ldots,i_r\}~{\rm{with}}~(u_j,v)\leq r\lfloor\frac{n}{r}\rfloor\}$.
Then it has $\cup_{j\in[e]}\{(u_j,v): v\in V_j\}=\{i\in\{i_1,i_2,\ldots,i_r\}: i\leq r\lfloor\frac{n}{r}\rfloor\}$.
Moreover, for each $a\in[0,\ell-1]$, define $w(a)=|\{i\in[e]: a_{u_i}\in V_i\}|$.
For each $0\leq j\leq e$, define $\Lambda_j=\{a\in[0,\ell-1]: w(a)=j\}$.
Then $\Lambda_0,\Lambda_1,\ldots,\Lambda_e$ form a partition of $[0,\ell-1]$. 
Recall the definition of $\mathcal{L}_s$ before Theorem \ref{thm-mds}, it has $\mathcal{L}_0, \ldots,\mathcal{L}_{m-\lfloor\frac{n}{r}\rfloor}$ form a partition of $[0,\ell-1]$. Also, for each $s\in[0,m-\lfloor\frac{n}{r}\rfloor]$, it has  that $\mathcal{L}_s$ has a partition $\mathcal{L}_s\cap \Lambda_j$ for $j=0,1,\ldots, e$.
Next we prove by induction on $s$ that for each $s\in[0,m-\lfloor\frac{n}{r}\rfloor]$, it has $\{x_{i,a}\}_{i\in[r],a\in\mathcal{L}_s}$ are all zeros.

At first, let $s=0$, we prove $\{x_{i,a}\}_{i\in[r],a\in\mathcal{L}_0}$ are all zeros.
If $U=\emptyset$, then it has $i_r>\cdots>i_1>r\lfloor\frac{n}{r}\rfloor$. 
For each $a\in \mathcal{L}_0$, by the $a$-th row of $(H_{t,i_1},\ldots,H_{t,i_r})\bm{x}^{\tau}=\bm{0}$, $t\in[r]$, it has $\sum_{j\in[r]}\lambda_{i_j}^{t-1}x_{j,a}=0$, $t\in[r]$, which implies $x_{j,a}=0$ for $j\in[r]$. Thus $\{x_{i,a}\}_{i\in[r],a\in\mathcal{L}_0}$ are all zeros.
If $U\neq\emptyset$ and $\Lambda_0=\emptyset$, 
then it must have that the $r$ nodes $\{i_1,i_2,\ldots,i_r\}$ belong to the same group $u$. For each $a\in\mathcal{L}_0$, fix the digits $a_{u'}$ with $u'\in[m]\setminus\{u\}$, consider the equations $(H_{t,i_1},\ldots,H_{t,i_r})\bm{x}^{\tau}=\bm{0}$, $t\in[r]$ labeled by $a(u,0), \ldots, a(u,r-1)$.
Then it has the same structure as that of code $\mathcal{C}_{\mathrm{II}}$ in Construction \ref{construction-2}.
In a similar way as in the proof of MDS property of $\mathcal{C}_{\mathrm{II}}$ (by letting $g=r$), one can obtain that $x_{i,a}=0$ for $i\in[r]$. 
Thus $\{x_{i,a}\}_{i\in[r],a\in\mathcal{L}_0}$ are all zeros.
If $U\neq\emptyset$ and $\Lambda_0\neq\emptyset$, 
we prove $\{x_{i,a}\}_{i\in[r],a\in\mathcal{L}_0\cap\Lambda_b}$ for all $b\in[0,e]$ are all zeros.
Firstly consider each $a\in\mathcal{L}_0\cap\Lambda_0$, according to the $a$-th equations of $(H_{t,i_1},\ldots,H_{t,i_r})\bm{x}^{\tau}=\bm{0}$, $t\in[r]$, it has $\sum_{j\in[r]}\lambda_{i_j}^{t-1}x_{j,a}=0$, $t\in[r]$, which implies $x_{j,a}=0$ for $j\in[r]$. Thus $\{x_{i,a}\}_{i\in[r],a\in\mathcal{L}_0\cap\Lambda_0}$ are all zeros.
Suppose for all $0\leq b'<b$, we have proved $\{x_{i,a}\}_{i\in[r],a\in\mathcal{L}_0\cap\Lambda_{b'}}$ are all zeros.
Now we prove $\{x_{i,a}\}_{i\in[r],a\in\mathcal{L}_0\cap\Lambda_{b}}$ are all zeros.
For each $a\in\mathcal{L}_0\cap\Lambda_{b}$, it has $a_u\in[2,r-1]$ for all $u\in[\lfloor\frac{n}{r}\rfloor+1,m]$, and 
there are exactly $b\leq e$ numbers in $U=\{u_1,\ldots, u_e\}$, say $u_1,u_2,\ldots,u_b$ without loss of generality, s.t. $a_{u_1}\in V_1, \ldots, a_{u_b}\in V_b$.
For convenience, here we denote $a_{u_1}=v_1\in V_1, \ldots, a_{u_b}=v_b\in V_b$, and denote $\bm{x}=(\bm{x_{i_1}},\ldots,\bm{x_{i_r}})$ if there is no ambiguity.
By the $a$-th equations of $(H_{t,i_1},\ldots,H_{t,i_r})\bm{x}^{\tau}=\bm{0}$, $t\in[r]$, it has

\begin{equation}\label{mds-eq1}
\begin{aligned}
&\sum_{z\in[b]}\Big\{\lambda_{(u_z-1)r+v_z+1}^{t-1}x_{(u_z-1)r+v_z+1, a}
+\sum_{ v\in V_z\setminus\{v_z\}}\lambda_{(u_z-1)r+v+1}^{t-1}(f(v_z,v) x_{(u_z-1)r+v_z+1, a(u_z,v)}+x_{(u_z-1)r+v+1, a})     \\
+&
\sum_{v\in[0,r-1]\setminus V_z} \lambda_{(u_z-1)r+v+1}^{t-1}\underline{x_{(u_z-1)r+v_z+1, a(u_z,v)}}f(v_z,v)  \Big\}   
+\sum_{\makecell[c]{\scriptsize j\in \{i_1,\ldots,i_r\}\\j>r\lfloor\frac{n}{r}\rfloor}}\lambda_j^{t-1}x_{j,a}=0,~~~~~~t\in[r],
\end{aligned}
\end{equation}
where $f(v_z,v)=1$ if $v_z<v$ and $f(v_z,v)=\tau$ if $v_z>v$. Note that $a(u_z,v)\in \Lambda_{b-1}$ for $v\in[0,r-1]\setminus V_z$, thus the underlined symbols in \eqref{mds-eq1} are all zeros by the hypothesis. 
Thus from \eqref{mds-eq1} one can obtain that the symbols
\begin{equation}\label{mds-eq2}
\begin{aligned}
&\{x_{(u_z-1)r+v_z+1, a}\}_{z\in[b]}
\cup\{\tau x_{(u_z-1)r+v_z+1, a(u_z,v)}+x_{(u_z-1)r+v+1, a}\}_{z\in[b],v\in V_z\setminus\{v_z\},v<v_z} \\
\cup&\{x_{(u_z-1)r+v_z+1, a(u_z,v)}+x_{(u_z-1)r+v+1, a}\}_{z\in[b],v\in V_z\setminus\{v_z\},v>v_z}
\cup\{x_{j,a}\}_{j\in \{i_1,\ldots,i_r\},j>r\lfloor\frac{n}{r}\rfloor}
\end{aligned}
\end{equation}
are all zeros.
Fix some $z\in[b]$, for each $v\in V_z\setminus\{v_z\}$, consider the $a'=a(u_z,v)$-th row of the equations  $(H_{t,i_1},\ldots,H_{t,i_r})\bm{x}^{\tau}=\bm{0}$, $t\in[r]$, one can similarlly solve out some zero symbols as in \eqref{mds-eq2}. Similar as in \eqref{eq-1-1}, using the zero symbol sums, one can compute that 
$x_{(u_z-1)r+v+1, a}=0$ for $v\in V_z\setminus\{v_z\}$. When $z$ run over all $z\in[b]$, one can obtain that $\{x_{(u_z-1)r+v+1, a}\}_{z\in[b],v\in V_z\setminus\{v_z\}}$ are all zeros.
Thus, replacing $\bm{x}=(\bm{x_{i_1}},\ldots,\bm{x_{i_r}})$ with $\bm{x}=(\bm{x_{1}},\ldots,\bm{x_{r}})$, one can obtain $\{x_{i,a}\}_{i\in[r],a\in\mathcal{L}_0\cap\Lambda_{b}}$ are all zeros.
Therefore, we have that $\{x_{i,a}\}_{i\in[r],a\in\mathcal{L}_0}$ are all zeros.

Now suppose that for all $0\leq s'<s$, we have proved the symbols $\{x_{i,a}\}_{i\in[r],a\in\mathcal{L}_{s'}}$ are all zeros.
Next we prove that $\{x_{i,a}\}_{i\in[r],a\in\mathcal{L}_{s}}$ are all zeros.

If $U=\emptyset$, then $i_r>\cdots>i_1>r\lfloor\frac{n}{r}\rfloor$. 
For each $a\in \mathcal{L}_s$, then $a$ has exactly $s$ digits in $(a_{\lfloor\frac{n}{r}\rfloor+1},\ldots, a_m)$ being $0$ or $1$, say $a_{\sigma_1},\ldots,a_{\sigma_s}$.
By the $a$-th row of $(H_{t,i_1},\ldots,H_{t,i_r})\bm{x}^{\tau}=\bm{0}$, $t\in[r]$, it has 
\begin{equation}\label{mds-eq3}
\sum_{j\in[r]}\lambda_{i_j}^{t-1}x_{j,a}
+\sum_{j\in[r]}\sum_{z\in[s]}\chi(a,\sigma_z,i_j)\sum_{w\in[2,r-1]}\gamma_{w-1}^{t-1}x_{j,a(\sigma_z,w)}
=0, ~~~~t\in[r],
\end{equation}
where 
$\chi(a,\sigma_z,i_j)=1$ if $a_{\sigma_z}=0,\sigma_z\in\Omega_{i_j,0}$ or $a_{\sigma_z}=1,\sigma_z\in\Omega_{i_j,1}$, otherwise $\chi(a,\sigma_z,i_j)=0$, here $\Omega_{i_j,0}, \Omega_{i_j,1}$ are defined in Subsection \ref{sec1-notation}.
Note that for $z\in[s], w\in[2,r-1]$, it has $a(\sigma_z,w)\in \mathcal{L}_{s-1}$, then $x_{j,a(\sigma_z,w)}$ for $j\in[r]$ are all zeros by the hypothesis.
Thus from \eqref{mds-eq3} one can compute that $x_{j,a}=0$, $j\in[r]$.

If $U\neq\emptyset$ and $\Lambda_0=\emptyset$, 
then $\{i_1,i_2,\ldots,i_r\}$ belong to the same group $u$. For each $a\in\mathcal{L}_s$, denote $a_u=j-1\in[0,r-1]$, and $a$ has exactly $s$ digits in $(a_{\lfloor\frac{n}{r}\rfloor+1},\ldots, a_m)$ being $0$ or $1$, say $a_{\sigma_1},\ldots,a_{\sigma_s}$.
By the $a$-th row of $(H_{t,i_1},\ldots,H_{t,i_r})\bm{x}^{\tau}=\bm{0}$, $t\in[r]$, it has 
\begin{equation}\label{mds-eq4}
\begin{aligned}
&\lambda_{i_j}^{t-1}x_{j,a}+\sum_{1\leq j'<j}\lambda_{i_{j'}}^{t-1}(x_{j',a}\!+\!\tau x_{j,a(u,j'-1)})+\sum_{j<j'\leq r}\lambda_{i_{j'}}^{t-1}(x_{j',a}+x_{j,a(u,j'-1)})   \\
+&\sum_{j\in[r]}\sum_{z\in[s]}\chi(a,\sigma_z,i_j)\sum_{w\in[2,r-1]}\gamma_{w-1}^{t-1}x_{j,a(\sigma_z,w)}
=0,  ~~~~~t\in[r].
\end{aligned} 
\end{equation}
where $\chi(a,\sigma_z,i_j)$ is defined in \eqref{mds-eq3}. 
Since $x_{j,a(\sigma_z,w)}=0$ for $j\in[r],z\in[s],w\in[2,r-1]$ by the hypothesis, then similarly as in the proof of MDS property of $\mathcal{C}_{\mathrm{II}}$ (by letting $g=r$), one can obtain $x_{i,a}=0$ for $i\in[r]$. Thus $\{x_{i,a}\}_{i\in[r],a\in\mathcal{L}_s}$ are all zeros.

If $U\neq\emptyset$ and $\Lambda_0\neq\emptyset$, 
we prove $\{x_{i,a}\}_{i\in[r],a\in\mathcal{L}_s\cap\Lambda_b}$ for all $b\in[0,e]$ are all zeros.
At first, for each $a\in\mathcal{L}_s\cap\Lambda_0$, $a$ has exactly $s$ digits in $(a_{\lfloor\frac{n}{r}\rfloor+1},\ldots, a_m)$ being $0$ or $1$, say $a_{\sigma_1},\ldots,a_{\sigma_s}$.
According to the $a$-th equations of $(H_{t,i_1},\ldots,H_{t,i_r})\bm{x}^{\tau}=\bm{0}$, $t\in[r]$, it has 
$$
\sum_{j\in[r]}\lambda_{i_j}^{t-1}x_{j,a}
+\sum_{j\in[r]}\sum_{z\in[s]}\chi(a,\sigma_z,i_j)\sum_{w\in[2,r-1]}\gamma_{w-1}^{t-1}x_{j,a(\sigma_z,w)}=0, ~~t\in[r],
$$ 
where $\chi(a,\sigma_z,i_j)$ is defined in \eqref{mds-eq3}. 
Then one can compute $x_{j,a}=0$ for $j\in[r]$, since $x_{j,a(\sigma_z,w)}=0$ for $j\in[r],z\in[s],w\in[2,r-1]$ by the hypothesis. Thus $\{x_{i,a}\}_{i\in[r],a\in\mathcal{L}_s\cap\Lambda_0}$ are all zeros.
Suppose for all $0\leq b'<b$, we have proved $\{x_{i,a}\}_{i\in[r],a\in\mathcal{L}_s\cap\Lambda_{b'}}$ are all zeros.
Now we prove $\{x_{i,a}\}_{i\in[r],a\in\mathcal{L}_s\cap\Lambda_{b}}$ are all zeros.
For each $a\in\mathcal{L}_s\cap\Lambda_{b}$, then $a$ has exactly $s$ digits in $(a_{\lfloor\frac{n}{r}\rfloor+1},\ldots, a_m)$ being $0$ or $1$, say $a_{\sigma_1},\ldots,a_{\sigma_s}$, and 
there are exactly $b\leq e$ numbers in $U=\{u_1,\ldots, u_e\}$, say $u_1,u_2,\ldots,u_b$ without loss of generality, s.t. $a_{u_1}\in V_1, \ldots, a_{u_b}\in V_b$.
For convenience, denote $a_{u_1}=v_1, \ldots, a_{u_b}=v_b$, and $\bm{x}=(\bm{x_{i_1}},\ldots,\bm{x_{i_r}})$.
By the $a$-th equations of $(H_{t,i_1},\ldots,H_{t,i_r})\bm{x}^{\tau}=\bm{0}$, $t\in[r]$, it has
\begin{equation}\label{mds-eq5}
\begin{aligned}
&\sum_{z\in[b]}\Big\{\lambda_{(u_z-1)r+v_z+1}^{t-1}x_{(u_z-1)r+v_z+1, a}
+\sum_{ v\in V_z\setminus\{v_z\}}\lambda_{(u_z-1)r+v+1}^{t-1}(f(v_z,v) x_{(u_z-1)r+v_z+1, a(u_z,v)}+x_{(u_z-1)r+v+1, a})     \\
+&
\sum_{v\in[0,r-1]\setminus V_z} \lambda_{(u_z-1)r+v+1}^{t-1}x_{(u_z-1)r+v_z+1, a(u_z,v)}f(v_z,v)  \Big\}   
+\sum_{\makecell[c]{\scriptsize j\in \{i_1,\ldots,i_r\}\\j>r\lfloor\frac{n}{r}\rfloor}}\lambda_j^{t-1}x_{j,a}   \\
+&\sum_{j\in[r]}\sum_{z\in[s]}\chi(a,\sigma_z,i_j)\sum_{w\in[2,r-1]}\gamma_{w-1}^{t-1}x_{j,a(\sigma_z,w)}
=0,~~~~~t\in[r],
\end{aligned}
\end{equation}
where $f(v_z,v)$ is defined in \eqref{mds-eq1} and $\chi(a,\sigma_z,i_j)$ is defined in \eqref{mds-eq3}. Note that $x_{j,a(\sigma_z,w)}=0$ for $j\in[r],z\in[s],w\in[2,r-1]$ by the hypothesis. Then in a similar way as in the case of $s=0$, one can compute that $\{x_{i,a}\}_{i\in[r],a\in\mathcal{L}_s}$ are all zeros.
This completes the proof.

\end{proof}

\subsection{Optimal-access property}\label{repair-h2}
We show the optimal access property of $\mathcal{C}$ in Construction \ref{construction-3}.
For some $u\in[m]$ and $v\in[0,r-1]$, define $A(u,v)=\{a\in[0,\ell-1]: a_u=v\}$.
Suppose the two nodes $\{i_1,i_2\}$ with $i_1<i_2$ are erased, and 
recall the definition $\mathcal{P}$ and map $\pi$ in Subsection \ref{sec1-notation}.
As illustrated in Remark \ref{remark},
If the erased node pair $(i_1,i_2)\in\mathcal{P}_u$ for some $u\in[\lfloor\frac{n}{r}\rfloor]$, then $\pi(i_1,i_2)=u$. Write $i_j=(u-1)r+v_j+1$ for $j\in[2]$, we use parity check equations with rows labeled by $a\in A(u,v_1)$ (resp. $a\in A(u,v_2)$) for repair of node $i_1$ (resp. $i_2$).
If the erased node pair $(i_1,i_2)\in\mathcal{P}_0$, denote $\rho=\pi(i_1,i_2)$, we use parity check equations with rows labeled by $a\in A(\rho,0)$ (resp. $a\in A(\rho,1)$) for repair of node $i_1$ (resp. $i_2$).
In the following, we illustrate the precise repair process of the two kinds of repair patterns in Theorem \ref{thm1-repair} and Theorem \ref{thm2-repair}, respectively.
\begin{theorem}\label{thm1-repair}
Suppose the erased nodes $(i_1,i_2)\in\mathcal{P}_u$ for some $u\in[\lfloor\frac{n}{r}\rfloor]$. Write $i_j=(u-1)r+v_j+1$ with $v_j\in[0,r-1]$ for $j\in[2]$. The repair process includes the following two phases.

 \begin{itemize}
\item (Download phase)
Node $i_1$ downloads $\{c_{p,a}: a\in A(u,v_1)\}$ from each helper node $p\in[n]\setminus\{i_1,i_2\}$, and
node $i_2$ downloads $\{c_{p,a}: a\in A(u,v_2)\}$ from each helper node $p\in[n]\setminus\{i_1,i_2\}$.
\item (Collaboration phase) Node $i_1$ recursively computes and transmits data $\{c_{i_1,a(u,v_2)}+c_{i_2,a}\}_{a\in A(u,v_1)}$ to node $i_2$, and node $i_2$ recursively computes and transmits data $\{c_{i_1,a}+\tau c_{i_2,a(u,v_1)}\}_{a\in A(u,v_2)}$ to node $i_1$.
\end{itemize}
\end{theorem}

\begin{proof}

Recall $\mathcal{L}_s$ defined in Subsection \ref{mds-h2}, it has that
$\mathcal{L}_0, \ldots,\mathcal{L}_{m-\lfloor\frac{n}{r}\rfloor}$ form a partition of $[0,\ell-1]$.
We prove by induction on $s$ that for each $s\in[0,m-\lfloor\frac{n}{r}\rfloor]$, after download phase and collaboration phase,
node $i_1$ and $i_2$ can recover $\{c_{i_1,a}: a\in\mathcal{L}_s\}$ and $\{c_{i_2,a}: a\in\mathcal{L}_s\}$ respectively.

At first, let $s=0$.
We firstly consider data recovery of node $i_1$. 
For each $a\in \mathcal{L}_0\cap A(u,v_1)$, it has $a_u=v_1$ and $a_{\lfloor\frac{n}{r}\rfloor+1},\ldots, a_m\in[2,r-1]$.
According to the $a$-th row of parity check equations $(H_{t,1},\ldots,H_{t,n})\bm{c}^{\tau}=\bm{0}$, $t\in[r]$, it has
\begin{equation}\label{eq1-repairthm1}
\begin{aligned}
&\lambda_{i_1}^{t-1}c_{i_1,a}+\sum_{v=0}^{v_1-1}\lambda_{(u-1)r+v+1}^{t-1}(c_{(u-1)r+v+1,a}+\tau c_{i_1,a(u,v)})+
\sum_{v=v_1+1}^{r-1}\lambda_{(u-1)r+v+1}^{t-1}(c_{(u-1)r+v+1,a}+c_{i_1,a(u,v)})     \\
+& 
\sum_{j\in[n]\setminus[(u-1)r+1,ur]}\lambda_j^{t-1}c_{j,a} 
+\sum_{u'\in[\lfloor\frac{n}{r}\rfloor]\setminus\{u\}} \sum_{v\in[0,r-1]\setminus\{a_{u'}\}}\lambda_{(u'-1)r+v+1}^{t-1} c_{(u'-1)r+a_{u'}+1,a(u',v)}f(a_{u'},v)
  =0, ~~t\in[r],
\end{aligned}
\end{equation}
where $f(a_{u'},v)=1$ if $a_{u'}<v$ and $f(a_{u'},v)=\tau$ if $a_{u'}>v$.
Note that $a(u',v)\in A(u,v_1)$ for $u'\in[\lfloor\frac{n}{r}\rfloor]\setminus\{u\}$, then 
the symbols $c_{(u'-1)r+a_{u'}+1,a(u',v)}$ in \eqref{eq1-repairthm1} are downloaded data. Moreover, $c_{j,a}$ for $j\neq i_1,i_2$ are also downloaded and recall $i_2=(u-1)r+v_2+1$ where $v_1<v_2\leq r-1$.
Then node $i_1$ can recover from \eqref{eq1-repairthm1}  the following set of $r$ symbols
$$
\{c_{i_1,a(u,v)} : v\in[0,r-1]\setminus\{v_2\}\}\cup\{c_{i_1,a(u.v_2)}+c_{i_2,a}\}
$$
When $a$ runs over all integers in $\mathcal{L}_0\cap A(u,v_1)$, node $i_1$ can obtain 
$$
\{c_{i_1,a(u,v_2)}+c_{i_2,a}, ~c_{i_1,a(u,v)} : v\in[0,r-1]\setminus\{v_2\}\}_{a\in \mathcal{L}_0\cap A(u,v_1)}.
$$
In a similar way, by considering parity check equations labeled by $a\in \mathcal{L}_0\cap A(u,v_2)$, node $i_2$ can recover the symbols
$$
\{c_{i_1,a}+\tau c_{i_2,a(u,v_1)}, ~c_{i_2,a(u,v)} : v\in[0,r-1]\setminus\{v_1\}\}_{a\in \mathcal{L}_0\cap A(u,v_2)}.
$$
   Then by exchanging the data $\{c_{i_1,a(u,v_2)}+c_{i_2,a}\}_{a\in \mathcal{L}_0\cap A(u,v_1)}$ and $\{c_{i_1,a}+\tau c_{i_2,a(u,v_1)}\}_{a\in \mathcal{L}_0\cap A(u,v_2)}$ between the two nodes, node $i_1$ and node $i_2$ can recover $\{c_{i_1,a}\}_{a\in\mathcal{L}_0}$ and $\{c_{i_2,a}\}_{a\in\mathcal{L}_0}$, respectively.

Now suppose node $i_1$ and node $i_2$ has respectively recovered $\cup_{s'\leq s-1}\{c_{i_1,a}: a\in\mathcal{L}_{s'}\}$ and $\cup_{s'\leq s-1}\{c_{i_2,a}: a\in\mathcal{L}_{s'}\}$.
We prove the case $s'=s$ that node $i_1$ and node $i_2$ can respectively recover $\{c_{i_1,a}\}_{a\in\mathcal{L}_s}$ and $\{c_{i_2,a}\}_{a\in\mathcal{L}_s}$.
Let us firstly consider the recovery of node $i_1$.
For each $a\in\mathcal{L}_{s}\cap A(u,v_1)$, it has $a_u=v_1$ and $a$ has exactly $s$ digits in $(a_{\lfloor\frac{n}{r}\rfloor+1},\ldots, a_m)$ being $0$ or $1$, say $a_{\sigma_1},\ldots,a_{\sigma_s}$. By the $a$-th row of parity check equations $(H_{t,1},\ldots,H_{t,n})\bm{c}^{\tau}=\bm{0}$, $t\in[r]$, it has 
\begin{equation}\label{repair-eq2}
\begin{aligned}
&\lambda_{i_1}^{t-1}c_{i_1,a}+\sum_{v=0}^{v_1-1}\lambda_{(u-1)r+v+1}^{t-1}(c_{(u-1)r+v+1,a}+\tau c_{i_1,a(u,v)})+
\sum_{v=v_1+1}^{r-1}\lambda_{(u-1)r+v+1}^{t-1}(c_{(u-1)r+v+1,a}+c_{i_1,a(u,v)})     \\
+& 
\sum_{j\in[n]\setminus[(u-1)r+1,ur]}\lambda_j^{t-1}c_{j,a} 
+\sum_{u'\in[\lfloor\frac{n}{r}\rfloor]\setminus\{u\}} \sum_{v\in[0,r-1]\setminus\{a_{u'}\}}\lambda_{(u'-1)r+v+1}^{t-1} c_{(u'-1)r+a_{u'}+1,a(u',v)}f(a_{u'},v) \\
+&\sum_{j\in[n]}\sum_{z\in[s]}\sum_{v\in[2,r-1]}\chi(a,\sigma_{z},j)\gamma_{v-1}^{t-1}c_{j,a(\sigma_{z},v)}=0,  ~~~t\in[r],
\end{aligned}
\end{equation}
where $f(a_{u'},v)=1$ if $a_{u'}<v$ and $f(a_{u'},v)=\tau$ if $a_{u'}>v$, and $\chi(a,\sigma_z,j)=1$ if $\sigma_z\in\Omega_{j,0}, a_{\sigma_z}=0$ or $\sigma_z\in\Omega_{j,1}, a_{\sigma_z}=1$, otherwise, $\chi(a,\sigma_z,j)=0$. 
Observe that $a(\sigma_{z},v)\in \mathcal{L}_{s-1}\cap A(u,v_1)$ for $z\in[s]$ and $v\in[2,r-1]$, then $c_{j,a(\sigma_{z},v)}$ is either downloaded data in the helper node (when $j\in[n]\setminus\{i_1,i_2\}$) or has been recovered by the hypothesis (when $j\in\{i_1,i_2\}$).
Thus from (\ref{repair-eq2}) node $i_1$ can compute the following set of $r$ symbols
$$
\{c_{i_1,a(u,v)} : v\in[0,r-1]\setminus\{v_2\}\}\cup\{c_{i_1,a(u.v_2)}+c_{i_2,a}\}.
$$
When $a$ runs over all integers in $\mathcal{L}_{s}\cap A(u,v_1)$, node $i_1$ can obtain 
$$
\{c_{i_1,a(u,v_2)}+c_{i_2,a}, ~c_{i_1,a(u,v)} : v\in[0,r-1]\setminus\{v_2\}\}_{a\in \mathcal{L}_{s}\cap A(u,v_1)}.
$$
In a similar way, by considering parity check equations labeled by $a\in \mathcal{L}_{s}\cap A(u,v_2)$, node $i_2$ can recover the symbols
$$
\{c_{i_1,a}+\tau c_{i_2,a(u,v_1)}, ~c_{i_2,a(u,v)} : v\in[0,r-1]\setminus\{v_1\}\}_{a\in \mathcal{L}_{s}\cap A(u,v_2)}.
$$
   Then by exchanging the data $\{c_{i_1,a(u,v_2)}+c_{i_2,a}\}_{a\in \mathcal{L}_s\cap A(u,v_1)}$ and $\{c_{i_1,a}+\tau c_{i_2,a(u,v_1)}\}_{a\in \mathcal{L}_s\cap A(u,v_2)}$ between the two nodes, node $i_1$ and node $i_2$ can recover $\{c_{i_1,a}\}_{a\in\mathcal{L}_s}$ and $\{c_{i_2,a}\}_{a\in\mathcal{L}_s}$, respectively.
Thus the two nodes can be cooperatively recovered.

Moreover, 
in the collaboration phase, the total data node $i_1$ transmitted to node $i_2$ are
$$\cup_{s\in[0,m-\lfloor\frac{n}{r}\rfloor]}\{c_{i_1,a(u,v_2)}+c_{i_2,a}: a\in \mathcal{L}_s\cap A(u,v_1)\}=\{c_{i_1,a(u,v_2)}+c_{i_2,a}: a\in A(u,v_1)\}.$$
The data node $i_2$ transmitted to node $i_1$ are
$$\cup_{s\in[0,m-\lfloor\frac{n}{r}\rfloor]}\{c_{i_1,a}+\tau c_{i_2,a(u,v_1)}: a\in \mathcal{L}_s\cap A(u,v_2)\}=\{c_{i_1,a}+\tau c_{i_2,a(u,v_1)}: a\in A(u,v_2)\}.$$
Thus the total amount of data communicated in the collaboration phase is $\frac{2\ell}{r}$ symbols. Moreover, the amount of data downloaded and accessed in the download phase are both $\frac{2(n-2)\ell}{r}$ symbols. Then the total repair bandwidth is $\frac{2(n-1)\ell}{r}$ symbols, achieving the cut-set bound. 
\end{proof}

\begin{theorem}\label{thm2-repair}
Suppose the erased nodes $(i_1,i_2)\in P_0$, denote $\rho=\pi(i_1,i_2)$.
Then the repair process includes the following two phases.
 \begin{itemize}
\item (Download phase)
(1) Node $i_1$ downloads $\{c_{p,a}: a\in A(\rho,0)\}$ from each helper node $p\in[n]\setminus\{i_1,i_2\}$,
and can recover the data $\{c_{i_1,a}\}_{a\in[0,\ell-1]\setminus A(\rho,1)}\cup\{c_{i_2,a}\}_{a\in A(\rho,0)}$.
(2) Node $i_2$ downloads $\{c_{p,a}: a\in A(\rho,1)\}$ from each helper node $p\in[n]\setminus\{i_1,i_2\}$,
and can recover the data $\{c_{i_2,a}\}_{a\in[0,\ell-1]\setminus A(\rho,0)}\cup\{c_{i_1,a}\}_{a\in A(\rho,1)}$.
\item (Collaboration phase)
Node $i_1$ transmits $\{c_{i_2,a}\}_{a\in A(\rho,0)}$ to node $i_2$, and node $i_2$ transmits $\{c_{i_1,a}\}_{a\in A(\rho,1)}$ to node $i_1$.
\end{itemize}
\end{theorem}

\begin{proof}
We prove (1) in the download phase, then (2) follows similarly, and the collaboration phase and node recovery will be straightforward.
Recall the definition of $\mathcal{L}_s$ in Subsection \ref{mds-h2},
it has
$\mathcal{L}_0, \ldots,\mathcal{L}_{m-\lfloor\frac{n}{r}\rfloor}$ form a partition of $[0,\ell-1]$.
Then $A(\rho,0)\cap\mathcal{L}_s$, for $s\in[m-\lfloor\frac{n}{r}\rfloor]$ form a partition of $A(\rho,0)$.
Observe that 
$$\{c_{i_1,a}\}_{a\in[0,\ell-1]\setminus A(\rho,1)}\cup\{c_{i_2,a}\}_{a\in A(\rho,0)}=\{c_{i_1,a},c_{i_2,a},c_{i_1,a(\rho,2)},\dots,c_{i_1,a(\rho,r-1)}\}_{a\in A(\rho,0)}.$$
We prove by induction on $s$ that for each $s\in[m-\lfloor\frac{n}{r}\rfloor]$,
node $i_1$ can recover $\{c_{i_1,a},c_{i_2,a},c_{i_1,a(\rho,2)},\dots,c_{i_1,a(\rho,r-1)}\}_{a\in A(\rho,0)\cap\mathcal{L}_s}$ using the downloaded data $\{c_{p,a}: a\in A(\rho,0)\}_{p\in[n]\setminus\{i_1,i_2\}}$. 

At first, let $s=1$.
For each $a\in A(\rho,0)\cap\mathcal{L}_1$, it has $a_{\rho}=0$ and $a_j\in[2,r-1]$ for all $j\in[\lfloor\frac{n}{r}\rfloor+1,m]\setminus\{\rho\}$.
According to the $a$-th row of parity check equations $(H_{t,1},\ldots,H_{t,n})\bm{c}^{\tau}=\bm{0}$, $t\in[r]$, it has 
\begin{equation}\label{eq1-thm4}
\sum_{j\in[n]}\lambda_{j}^{t-1}c_{j,a}+
\sum_{u\in[\lfloor\frac{n}{r}\rfloor]}\sum_{v\in[0,r-1]\setminus\{a_u\}}\lambda_{(u-1)r+v+1}^{t-1}\underline{c_{(u-1)r+a_u+1,a(u,v)}} f(a_u,v) 
+\sum_{w=2}^{r-1}\gamma_{w-1}^{t-1}c_{i_1,a(\rho,w)}=0, ~~~t\in[r],
\end{equation}
where $f(a_u,v) =1$ if $a_u<v$ and $f(a_u,v) =\tau$ if $a_u>v$.
Next we claim that the underlined symbols $c_{(u-1)r+a_u+1,a(u,v)}$'s in \eqref{eq1-thm4} can be obtained by node $i_1$.
If $i_1>r\lfloor\frac{n}{r}\rfloor$, then $i_2>i_1>r\lfloor\frac{n}{r}\rfloor$, and $c_{(u-1)r+a_u+1,a(u,v)}$'s with $u\in[\lfloor\frac{n}{r}\rfloor]$,$v\in[0,r-1]\setminus\{a_u\}$ in \eqref{eq1-thm4} are downloaded data in the helper nodes since $a(u,v)\in A(\rho,0)\cap\mathcal{L}_1$.
If $i_1\leq r\lfloor\frac{n}{r}\rfloor$ and $i_2>r\lfloor\frac{n}{r}\rfloor$, then we can write $i_1=(u_1-1)r+v_1+1$ with some $u_1\in[\lfloor\frac{n}{r}\rfloor]$ and $v_1\in[0,r-1]$.
Firstly consider each $a\in A(\rho,0)\cap\mathcal{L}_1$ with $a_{u_1}\neq v_1$, the underlined symbols $c_{(u-1)r+a_u+1,a(u,v)}$'s in \eqref{eq1-thm4} are downloaded data. Along with the downloaded data $c_{j,a}, j\in[n]\setminus\{i_1,i_2\}$, node $i_1$ can solve from \eqref{eq1-thm4}  the unknown data $\{c_{i_1,a},c_{i_2,a},c_{i_1,a(\rho,2)},\dots,c_{i_1,a(\rho,r-1)}\}$ for all $a\in A(\rho,0)\cap\mathcal{L}_1$ with $a_{u_1}\neq v_1$.
Then consider each $a\in A(\rho,0)\cap\mathcal{L}_1$ with $a_{u_1}=v_1$.
In \eqref{eq1-thm4} when $u=u_1$, the underlined symbols $c_{(u_1-1)r+a_{u_1}+1,a(u_1,v)}=c_{i_1,a(u_1,v)}$ for $v\in[0,r-1]\setminus\{v_1\}$ have been previously solved out by node $i_1$.
When $u\in[\lfloor\frac{n}{r}\rfloor]\setminus\{u_1\}$, the symbols $c_{(u-1)r+a_{u}+1,a(u,v)}$ for $v\in[0,r-1]\setminus\{a_u\}$ are still downloaded data.
If $i_1<i_2\leq r\lfloor\frac{n}{r}\rfloor$, write $i_j=(u_j-1)r+v_j+1$ with $u_j\in[\lfloor\frac{n}{r}\rfloor]$, $v_j\in[0,r-1]$ for $j\in[2]$.
Firstly, consider each $a\in A(\rho,0)\cap\mathcal{L}_1$ with $a_{u_1}\neq v_1,a_{u_2}\neq v_2$, then the underlined symbols in \eqref{eq1-thm4}  are downloaded data. Node $ i_1$ can solve from \eqref{eq1-thm4}  the symbols $\{c_{i_1,a},c_{i_2,a},c_{i_1,a(\rho,2)},\dots,c_{i_1,a(\rho,r-1)}\}$. Then consider each $a\in A(\rho,0)\cap\mathcal{L}_1$ with $a_{u_1}\neq v_1,a_{u_2}=v_2$ or $a_{u_1}=v_1,a_{u_2}\neq v_2$.
Then the symbols $c_{(u-1)r+a_{u}+1,a(u,v)}$'s in \eqref{eq1-thm4} are either downloaded data at helper nodes or failed data at node $i_1$ or $i_2$ which have been previously recovered. Also, node $i_1$ can further recover from \eqref{eq1-thm4} the unknown symbols $\{c_{i_1,a},c_{i_2,a},c_{i_1,a(\rho,2)},\dots,c_{i_1,a(\rho,r-1)}\}$ for all $a\in A(\rho,0)\cap\mathcal{L}_1$ with $a_{u_1}\neq v_1,a_{u_2}=v_2$ or $a_{u_1}=v_1,a_{u_2}\neq v_2$.
At last, consider each $a\in A(\rho,0)\cap\mathcal{L}_1$ with $a_{u_1}=v_1,a_{u_2}=v_2$, then it is easy to see that the symbols $c_{(u-1)r+a_{u}+1,a(u,v)}$'s are either downloaded data  or have been previously recovered.
That is, the underlined data in \eqref{eq1-thm4} for all $a\in A(\rho,0)\cap\mathcal{L}_1$ can be obtained by node $i_1$. Thus using the downloaded data node $i_1$ can recover the following symbols
\[\{c_{i_1,a},c_{i_2,a},c_{i_1,a(\rho,2)},\dots,c_{i_1,a(\rho,r-1)}\}_{a\in A(\rho,0)\cap\mathcal{L}_1}.\]

Now suppose that node $i_1$ has recovered $\{c_{i_1,a},c_{i_2,a},c_{i_1,a(\rho,2)},\dots,c_{i_1,a(\rho,r-1)}\}_{a\in A(\rho,0)\cap\mathcal{L}_{s'}}$ for all $s'\leq s-1$. We prove the case $s'=s$ that $\{c_{i_1,a},c_{i_2,a},c_{i_1,a(\rho,2)},\dots,c_{i_1,a(\rho,r-1)}\}_{a\in A(\rho,0)\cap\mathcal{L}_{s}}$ can be recovered.

For each $a\in A(\rho,0)\cap\mathcal{L}_s$, it has $a_{\rho}=0$ and there exist exactly $s-1$ digits in $(a_{\lfloor\frac{n}{r}\rfloor+1},\ldots, a_{\rho-1},a_{\rho+1},\ldots,a_m)$ being $0$ or $1$, say $a_{\sigma_1},\ldots,a_{\sigma_{s-1}}$.
By the $a$-th row of parity check equations $(H_{t,1},\ldots,H_{t,n})\bm{c}^{\tau}=\bm{0}$, $t\in[r]$, it has
\begin{equation}\label{eq2-thm4}
\begin{aligned}
&\sum_{j\in[n]}\lambda_{j}^{t-1}c_{j,a}+
\sum_{u\in[\lfloor\frac{n}{r}\rfloor]}\sum_{v\in[0,r-1]\setminus\{a_u\}}\lambda_{(u-1)r+v+1}^{t-1}c_{(u-1)r+a_u+1,a(u,v)} f(a_u,v)  \\
+&\sum_{w=2}^{r-1}\gamma_{w-1}^{t-1}c_{i_1,a(\rho,w)}
+\sum_{j\in[n]}\sum_{z\in[s-1]}\chi(a,\sigma_z,j)\sum_{w\in[2,r-1]}\gamma_{w-1}^{t-1}c_{j,a(\sigma_z,w)}=0, ~~~t\in[r],
\end{aligned}
\end{equation}
where $f(a_u,v)$ and $\chi(a,\sigma_z,j)$ are defined in \eqref{repair-eq2}.
Since $a(\sigma_z,w)\in A(\rho,0)\cap\mathcal{L}_{s-1}$ for $z\in[s-1]$ and $w\in[2,r-1]$, then $c_{j,a(\sigma_z,w)}$ for $j\in[n]$ are either downloaded data at the helper nodes or have been previously recovered by the hypothesis.
Furthermore, similar as in the case of $s=1$, the symbols $c_{(u-1)r+a_u+1,a(u,v)}$'s with $u\in[\lfloor\frac{n}{r}\rfloor]$,$v\in[0,r-1]\setminus\{a_u\}$ in \eqref{eq2-thm4} can be obtained by node $i_1$. Then along with the downloaded data $c_{j,a}, j\in[n]\setminus\{i_1,i_2\}$, node $i_1$ can recover the following unknowns
\[\{c_{i_1,a},c_{i_2,a},c_{i_1,a(\rho,2)},\dots,c_{i_1,a(\rho,r-1)}\}_{a\in A(\rho,0)\cap\mathcal{L}_s}.\]
Therefore, node $i_1$ can recover the data $\{c_{i_1,a}\}_{a\in[0,\ell-1]\setminus A(\rho,1)}\cup\{c_{i_2,a}\}_{a\in A(\rho,0)}$ by downloading $\{c_{p,a}: a\in A(\rho,0)\}_{p\in[n]\setminus\{i_1,i_2\}}$.
Moreover, it is easy to verify that the cooperative repair satisfies optimal access property.
This completes the proof.

\end{proof}

\section{Conclusion}\label{conclusion}
We present an optimal-access cooperative MSR code with two erasures which has smaller sub-packetization than previous works. 
Despite the code still has exponential sub-packetization, it is possible to use this code and another large-distance code to build $\epsilon$-cooperative MSR codes with small data access and small sub-packetization as in \cite{Devi2019}.
An interesting problem in the future is to construct codes with any $h>2$ erasures and any $d$ helper nodes, and establish a lower bound on the sub-packetization of cooperative MSR codes.



%

\end{document}